\documentclass[11pt]{article}
\pdfoutput=1
\usepackage{graphicx}
\usepackage{epsfig}
\usepackage{amsmath}
\usepackage{amssymb}
\usepackage{amstext}
\usepackage{amsthm}
\usepackage{subfigure}
\usepackage{verbatim}
\usepackage{url}
\usepackage{color}
\usepackage{verbatim}
\usepackage{cite}

{\bfseries}{\itshape}
\newtheorem{lemma}{Lemma}{\bfseries}{\itshape}

\newcommand{\keywords}[1]{\par\addvspace\baselineskip
\noindent\keywordname\enspace\ignorespaces#1}

\let\oldproof\proof

\title{Testing Simultaneous Planarity when the Common Graph is 2-Connected}

\author{Bernhard Haeupler
\thanks{CSAIL, MIT, Cambridge, MA. Email:{\tt haeupler@mit.edu}}
\and Krishnam Raju Jampani \thanks{David R. Cheriton School of Computer Science, University of Waterloo, Waterloo, ON,  Email:{\tt krjampan@uwaterloo.ca}}
\and Anna Lubiw \thanks{David R. Cheriton School of Computer Science, University of Waterloo, Waterloo, ON, Email:{\tt alubiw@uwaterloo.ca}} }

\begin{document}
\DeclareGraphicsExtensions{.pdf, .png, .gif, .jpg}

\pagestyle{plain}
\pagenumbering{arabic}
\maketitle

\newcommand{\mr}{\ensuremath{\mathbf{\eta}}}
\newcommand{\len}{\ensuremath{t}}
\newcommand{\minor}{\ensuremath{\preceq}}
\newcommand{\lrangle}[1]{\ensuremath{\langle #1 \rangle }}
\newcommand{\floor}[1]{\ensuremath{\lfloor #1 \rfloor }}
\newcommand{\ceil}[1]{\ensuremath{\lceil #1 \rceil }}
\newcommand{\cart}{\ensuremath{\,\Box\,}}
\newcommand{\brac}[1]{\ensuremath{\lbrace #1 \rbrace }}
\newcommand{\ignore}[1]{}

\begin{abstract}
Two planar graphs $G_1$ and $G_2$ sharing some vertices and edges are \emph{simultaneously planar} if they have planar drawings such that a shared vertex [edge] is represented by the same point [curve] in both drawings.
It is an open problem whether simultaneous planarity can be tested efficiently.  We give a linear-time algorithm to test simultaneous planarity when the two graphs share a 2-connected subgraph. 
Our algorithm extends to the case of $k$ planar graphs where each vertex [edge] is either common to all graphs or belongs to exactly one of them, and the common subgraph is 2-connected.

\end{abstract}

\section{Introduction}

Let $G_1 = (V_1,E_1)$ and $G_2 = (V_2,E_2)$ be two graphs 
sharing some vertices and edges. The simultaneous planar embedding problem asks whether
there exist planar embeddings for $G_1$ and $G_2$ such that, in the two embeddings, each vertex $v \in V_1 \cap V_2$ is
mapped to the same point and each edge $e \in E_1 \cap E_2$ is mapped to the same
curve.  
We show that this problem can be
solved efficiently when the common graph $(V_1 \cap V_2, E_1 \cap E_2)$ is 2-connected.
%In this case---and more generally, so long as the common graph is connected---two graphs $G_1$ and $G_2$ are simultaneously planar if and only if there are planar embeddings of $G_1$ and $G_2$ such that  the  cyclic orderings of common edges  around common vertices are the same in both embeddings~\cite{JS}.

% motivation
The study of planar graphs has a long history and has generated many deep results~\cite{NC,MT,NR}.
There is hope that some of the structure of planarity may carry over to simultaneous planarity. 
A possible analogy is with matroids, where optimization results carry over from one matroid to the intersection of two matroids~\cite{CCPS}. 
On a more practical note, simultaneous planar embeddings are valuable for visualization purposes when two related graphs need to be displayed.  
For example,  
the two graphs may represent different relationships on the same node set, 
or they may be the ``before'' and ``after'' versions of a graph that has changed over time.

Over the last few years there has been a lot of work on simultaneous planar embeddings~\cite{DGL,EK,FGJMS,FJKS,F,GJPSS,JS,ADFPR,ADBF+}, often under the name ``simultaneous embeddings with fixed edges''. We mention a few results here and give a more detailed description in the Background section below. A major open question is whether simultaneous planarity of two graphs can be tested in polynomial time. The problem seems to be right on the feasibility boundary. The problem is NP-complete for three graphs~\cite{GJPSS} and the version where the planar drawings are required to be straight-line is already NP-hard for two graphs and only known to lie in PSPACE~\cite{EGJ+}. On the other hand, several classes of (pairs of) graphs are known to always have simultaneous planar embeddings~\cite{EK,DGL,F,FJKS,JS} and there are efficient algorithms to test simultaneous planarity for some very restricted graph-classes: biconnected outerplanar graphs~\cite{FJKS}, the case where one graph has at most one cycle~\cite{FGJMS} or the case where the common graph is a star~\cite{ADFPR}. The last result me mention here is that the problem of whether a planar embedding of a subgraph can be extended to a planar embedding of the whole graph is equivalent to testing simultaneous planarity between a planar graph and a 3-connected graph, and a linear time algorithm for these problems was given by Angelini et al.~\cite{ADBF+}.

In this paper we give 
%This paper\footnote{preliminary version in ISAAC, 2010~\cite{HLJ-ISAAC}} %shows how to efficiently
%gives 
a linear time algorithm to 
test simultaneous planarity of any two graphs that share a 2-connected subgraph.
% and thus greatly extends the classes of graph pairs for which a testing algorithm is known. 
Our algorithm builds on the planarity testing algorithm of Haeupler and Tarjan~\cite{HT}, which in turn unifies the planarity
testing algorithms of Lempel-Even-Cederbaum~\cite{LEC}, Shih-Hsu~\cite{SH} and Boyer-Myrvold~\cite{BM}. 
%We create structures representing all planar embeddings of $G_1$ and $G_2$, starting with $G_1 - G_2$ and $G_2 - G_1$ and then 
%adding the vertices of $V_1 \cap V_2$ while maintaining consistency between the two structures. 
We show how to extend our algorithm 
to the case of $k$ graphs where each vertex [edge] is either common to all graphs or belongs to exactly one of them, and the common graph is 2-connected. Although we concentrate on the decision version of the problem, we also 
show that two simultaneous planar graphs have a simultaneous planar drawing in which one graph is straight-line planar, and each edge of the other graph is drawn as a polygonal line with at most $| V_1 \cap V_2 |$ bends.

Independently (and at the same time as our conference version~\cite{HLJ-ISAAC}), Angelini et al.~\cite{ADFPR} gave an efficient algorithm to test simultaneous planarity of any two graphs that share a 2-connected subgraph.  
%Simultaneously and independently of this work Angelini et al.~\cite{ADFPR} showed how to test simultaneous planarity of two graphs when the common graph is  2-connected. 
Their algorithm is based on SPQR-trees, takes linear time (in the final version~\cite{ADFPR-f}), and is restricted to the case where the two graphs have the same vertex set. Very recently Bl{\"a}sius and Rutter~\cite{BR} have generalized our idea of simultaneous PQ trees, and 
shown how to test simultaneous planarity of two biconnected graphs that share a connected subgraph.  This is a substantial strengthening of our result.

The paper is organized as follows: Section \ref{sec:background} gives more background and related work. In Section~\ref{section:PQ} we review and develop some techniques for PQ-trees, which are needed for our algorithm, and in 
Section~\ref{section:planar} we review the Haeupler-Tarjan planarity testing algorithm. 
Section~\ref{sec:simplanarity} contains our simultaneous planarity testing algorithm, including the extension to $k$ graphs, and the result 
on drawing two simultaneous planar graphs.
%Our simultaneous planarity testing algorithm is in Section~\ref{sec:simplanarity}.
%We also show that our algorithm can be extended to solve a generalization of simultaneous planarity for $k$ graphs, whose common graph is 2-connected.

\subsection{Background}\label{sec:background}

Versions of simultaneous planarity have received much attention in recent years. Brass et al.~\cite{BCD+} introduced the concept of \emph{simultaneous geometric embeddings} of a pair of graphs---these are planar straight-line drawings such that any common vertex is represented by the same point.  Note that a common edge will necessarily be represented by the same line segment.  
It is NP-hard to test if two graphs have simultaneous geometric embeddings~\cite{EGJ+}.  For other work on simultaneous geometric embeddings see~\cite{AGKN} and its references.
 
The generalization to planar drawings where edges are not necessarily drawn as straight line segments, but any common edge must be represented by  the same curve was introduced by Erten and Kobourov~\cite{EK} and called \emph{simultaneous embedding with consistent edges}.  
Most other papers follow the conference version of Erten and Kobourov's paper and use the term \emph{simultaneous embedding with fixed edges (SEFE)}.   
In our paper we use the more self-explanatory term ``simultaneous planar embeddings.''  A further justification for this nomenclature is that there are combinatorial conditions on a pair of planar embeddings that are equivalent to simultaneous planarity.
Specifically, J{\"u}nger and Schultz give a characterization in terms of ``compatible embeddings'' [Theorem 4 in~\cite{JS}].  Specialized to the case where the common graph is connected, their result says that 
two planar embeddings are simultaneously planar if and only if the cyclic orderings of common edges around common vertices are the same in both embeddings.

Several papers~\cite{EK,DGL,F} show that pairs of graphs from certain restricted classes always have simultaneous planar embeddings, the most general result being that any planar graph has a simultaneous planar embedding with any tree~\cite{F}. On the other hand, there is an example of two outerplanar graphs that have no simultaneous planar embedding~\cite{F}. The graphs that have simultaneous planar embeddings when paired with any other planar graph have been characterized~\cite{FJKS}. In addition, J{\"u}nger and Schultz~\cite{JS} characterize the common graphs that permit simultaneous planar embeddings no matter what pairs of planar graphs they occur in. 

There are efficient algorithms to test simultaneous planarity for pairs of biconnected outerplanar graphs~\cite{FJKS}, and for pairs consisting of a planar graph and a graph that is either 3-connected~\cite{ADBF+} or has at most one cycle~\cite{FGJMS}. Testing simultaneous planarity when the common graph is a tree was shown to be equivalent to a book-embedding problem which together with~\cite{HN} implies that testing simultaneous planarity when the common graph is a star can be done in linear time~\cite{ADFPR}.

There is another, even weaker form of simultaneous planarity, where common vertices must be represented by common points, but the planar drawings are otherwise completely independent, with edges drawn as Jordan curves. Any set of planar graphs can be represented this way by virtue of the result that a planar graph can be drawn with any fixed vertex locations~\cite{PW}. 

The idea of  ``simultaneous graph representations'' has also been applied to intersection representations.
For the case of interval graphs, the problem is to find representations for two graphs as intersection graphs of intervals in the real line, with the property that a common vertex is represented by the same interval in both representations.  This can be done in polynomial time~\cite{JL}. There are also polynomial time algorithms to find simultaneous representations of pairs of
chordal, comparability and permutation graphs~\cite{JL2}.

\section{PQ-trees}
\label{section:PQ}

% I've done minor editing in this section

Many planarity testing algorithms in the literature~\cite{HT,LEC,SH,BM} use PQ-trees (or a variation) to obtain a linear-time implementation. PQ-trees were discovered by Booth and Lueker~\cite{BL} and are used, not only for planarity testing, but for many other applications like recognizing interval graphs or testing matrices for the consecutive-ones property. 
We first review PQ-trees and then in Subsection~\ref{sec:int-PQ} we show how to manipulate pairs of PQ-trees. 

A PQ-tree represents the permutations of a set of elements satisfying a family of constraints. Each constraint 
specifies that a certain subset of elements must appear consecutively in any permutation. PQ-trees are rooted trees with 
internal nodes being labeled `P' or `Q', and are drawn using a circle or a rectangle (or double circles in \cite{HM}), respectively.
The leaves of a PQ-tree  correspond to the elements whose orders are represented. % editing to save space
PQ-trees are equivalent under arbitrary reorderings of the children of a P-node and reversals of the order of children of a Q-node. We consider a node with two children 
to be a Q-node.
A leaf-order of a PQ-tree is the order in which its leaves are visited in an in-order traversal of the tree. 
The set of permutations represented by a PQ-tree is the set of leaf-orders of equivalent PQ-trees. 
%The set of leaf-orders of PQ-trees that are equivalent to a given one gives the family 
%of permutations that this PQ-tree represents. 
Given a PQ-tree tree $T$ on a set $U$ of elements, adding a consecutivity
constraint on a set $S \subseteq U$,  {\em reduces} $T$ to a PQ-tree $T'$, such that the leaf-orders of $T'$ are 
precisely the leaf-orders of $T$ in which the elements of $S$ appear consecutively. Booth and Lueker~\cite{BL} gave an efficient implementation of PQ-trees that supports this operation in amortized $O(|S|)$ time. Their implementation furthermore allows an efficient complement-reduction, i.e., adding the constraint on a set $S$ when $U \setminus S$ is given in amortized $O(|U \setminus S|)$ time~\cite{HT}.

Although PQ-trees were invented to represent linear orders, they can be reinterpreted to represent circular orders as well~\cite{HT}:
Given a PQ-tree we imagine that there is a new special leaf $s$ attached as the ``parent'' of the root. A circular leaf order of the augmented tree is a circular order that begins at the special leaf, followed by a linear order of the remaining PQ-tree and ending at the special leaf. Again a PQ-tree represents all circular leaf-orders of equivalent PQ-trees. It is easy to see that a consecutivity constraint on such a set of circular orders directly corresponds to a consecutivity constraint on the original set of linear leaf-orders and a reduction on the circular orders corresponds to a standard or complement reduction. Note that using PQ-trees for circular orders requires solely this different view on PQ-trees but does not need any change in their implementation. As such it turns out that PC-trees introduced in~\cite{HM} are the exact same data structure as PQ-trees albeit with this circular interpretation.

\subsection{Intersection and projection of PQ-trees}
\label{sec:int-PQ}

In this section we develop simple techniques to obtain consistent orders from two PQ-trees. More precisely when two PQ-trees share some but not (necessarily) all leaves, we want to find a permutation represented by each of them with a consistent ordering on the shared leaves.  

Note that the set of such consistent orderings is not representable by a PQ-tree.  
For example, consider the ground set $\{1, 2, 3, 4\}$ and the constraint that  $\{2,3\}$ must be consecutive among $\{1,2,3\}$.  This is why we restrict our goal to finding one permutation of each PQ-tree such that the two permutations are consistent on the shared leaves.

The idea %for this 
is to first \emph{project} both PQ-trees to the common elements, intersect the resulting PQ-trees, pick one remaining order and finally ``lift'' this order back. We now describe the individual steps of this process in more detail.  

The {\em projection} of a PQ-tree on a subset of its leaves $S$ is a PQ-tree obtained by deleting all elements not
in $S$ and simplifying the resulting tree. Simplifying a tree means that we (recursively) delete any internal node
that has no children, and delete any node that has a single child by making the child's grandparent become its
parent. This can easily be implemented in linear time.

Given two PQ-trees on the same set of leaves (elements) we define their {\em intersection} to be the PQ-tree $T$ that represents exactly all orders that are leaf-orders in both trees. This intersection can be computed in polynomial time as follows. 

\begin{enumerate}
\item Initialize $T$ to be the first PQ-tree.
\item For each P-node in the second PQ-tree, reduce $T$ by adding a consecutivity constraint on all the P-node's descendant leaves.
\item For each Q-node in the second tree, and for each pair of adjacent children of it, reduce $T$ by adding a consecutivity
constraint on all the descendant leaves of the two children.
\end{enumerate}

Using the efficient PQ-tree implementation, computing such an intersection be sped up to amortized linear time in the size of the two PQ-trees. This is a relatively straight-forward extension of the PQ-tree reduction and described in detail in Booth's thesis~\cite{Booth}. 

These two operations are enough to achieve our goal. Given two PQ-trees $T_1$ and $T_2$ defined on different element (leaf) sets, 
we define $S$ to be the set of common elements. Now we first construct the projections of both PQ-trees on $S$ and then compute 
their intersection $T$ as described above. Any permutation of $S$ represented by $T$ can now easily be ``lifted'' back to 
permutations of $T_1$ and of $T_2$ that respect the chosen ordering of $S$. Furthermore, any two permutations of $T_1$ and 
$T_2$ that are consistent on $S$ can be obtained this way. 
%Note that it is not possible to efficiently represent the set of these permutations as a PQ-tree. 

We note that techniques to ``merge'' PQ-trees were also presented by J\"unger and Leipert~\cite{JLjgaa} in work on level planarity.
Their merge is conceptually and technically different from ours in that the result of their merge is a single PQ-tree
whereas we compute two orderings that are consistent on common elements. The set of all these consistent orderings cannot be captured by a PQ-tree.

\section{Planarity testing}
\label{section:planar}
\label{sec:planar}

In this Section, we review the recent algorithm of Haeupler and Tarjan~\cite{HT} for testing the planarity of a graph. Next
we extend it to an algorithm for testing simultaneous planarity. We first begin with some basic definitions.

Let $G=(V,E)$ be a connected graph on vertex set $V = \brac{v_1,\cdots,v_n}$ and let ${\cal O}$ be an ordering of the vertices of $V$.
An edge $v_iv_j$ is an \emph{in-edge} of $v_i$ (in ${\cal O}$) if $v_j$ appears before $v_i$ in ${\cal O}$, and  $v_iv_j$ is
an \emph{out-edge} of $v_i$ if $v_j$ appears after $v_i$ in ${\cal O}$.

    An st-ordering of $G$ is an ordering ${\cal O}$ of the vertices of $G$, such that the first vertex of ${\cal O}$ is adjacent
to the last vertex of ${\cal O}$ and every intermediate vertex has an in-edge and an out-edge. It is well-known that $G$ has an st-ordering
if and only if it is 2-connected. Further, an st-ordering can be computed in linear time~\cite{ET}.

%    A {\em combinatorial embedding} of $G$, denoted by ${\cal C}(G)$, is defined as a clockwise circular ordering of the incident 
%edges of $v_i$, for each $i \in \brac{1,\cdots,n}$, with respect to a planar drawing of $G$. 

%Comment Bernhard: This If $\cal C$ is a combinatorial embedding of $G$, we use ${\cal C}(v_i)$ to denote the clockwise circular ordering of edges incident with $v_i$ in $\cal C$.

%\subsection{Planarity testing using PQ-trees} \label{sec:pq-tree-planarity-testing}
%{\note you don't need subsection heading}

%Let $G=(V,E)$ be a connected graph. 

The planarity testing algorithm of Haeupler and Tarjan embeds vertices (and their adjacent edges) one at a time and
maintains all possible partial embeddings of the embedded subgraph at each stage. Edges with two, one or no endpoint(s) embedded are called (fully-)embedded, half-embedded and non-embedded respectively. For the correctness of the algorithm it is crucial that the vertices are added in a leaf-to-root order of a spanning-tree. This guarantees that, at any time, the remaining vertices induce a connected graph and hence lie in a single face of the partial embedding which can be thought of as the outer face. The crucial observation is that now (for the sake of extending this embedding) a partial embedding can be completely characterized by describing, for each connected component, the cyclic order of the half-embedded edges around it. Furthermore, PQ-trees can be used in a natural way to describe the set of all possible cyclic orders and thus all possible partial planar embeddings. Before describing how to update the PQ-trees corresponding to sets of partial embeddings when a vertex is added, we first discuss the choice of the specific embedding order we use. In general using either a leaf-to-root order of a depth-first spanning tree or an st-order leads to particularly simple implementations that run in linear-time. Indeed these two orders are essentially the only two orders in which the algorithm runs in linear-time using the standard PQ-tree implementation. Our simultaneous planarity algorithm will use a mixture of the two orders: We first add the vertices that are contained in only one of the graphs using a depth-first search order and then add the common vertices using an st-ordering. Note that if all vertices are common, both $G_1$ and $G_2$ will be 2-connected and have a common st-order. In this case we only use this st-order. We now give an overview of how the update steps of the planarity test work for each of these orderings. \\

\noindent
{\bf Leaf-to-root order of a depth-first spanning tree:} \\
Let $v_1,v_2,\cdots,v_n$ be a leaf-to-root order of a depth-first spanning tree of $G$. Note that at stage $i$, the
vertices $\brac{v_1, \cdots, v_i}$ may induce several components. We maintain a PQ-tree for each component representing the set of possible circular orderings of its out-edges. Using a depth-first spanning tree, in contrast to an arbitrary spanning tree, has the advantage that we can easily maintain the invariant that the edge to the smallest node greater than $i$ will be the special leaf. 

Adding $v_{i+1}$ can lead to merging several components into one. It is easy to see, e.g., in the left part of Figure \ref{FIG:step}, that after embedding $v_{i+1}$ the edges of each merging component that go to $v_{i+1}$ have to be consecutive since otherwise a half-embedded edge would be enclosed in the newly formed component. To go to the next stage, we thus first reduce each PQ-tree corresponding to such a component by adding a consecutivity constraint that requires the set of out-edges that are incident to $v_{i+1}$ to be consecutive. We then delete these edges since they are now fully embedded. By the invariant stated above the special leaf is among these edges. Note that the resulting PQ-tree for a component now represents the set of all possible linear-orders of the out-edges that are not incident to $v_{i+1}$. Now we construct the PQ-tree for the new merged component including $v_{i+1}$ as follows: 

Let $v_l$ be the parent of $v_{i+1}$ in the depth-first spanning tree. The PQ-tree for the new component consists of the edge $v_{i+1}v_l$ as the special leaf and a new P-node as a root and whose children are all the remaining out-edges of $v_{i+1}$ and the roots of the PQ-trees of the reduced components (similar to the picture in Figure \ref{FIG:setup}). Note that by choosing the edge $v_{i+1}v_l$ as the special leaf we again maintain the above mentioned invariant.

It is easy to verify that these operations capture exactly all possible circular orders of half-embedded edges around the new embedded connected component. Note that if the reduction step fails at any stage then the graph must be non-planar. Otherwise the algorithm concludes that the graph is planar. \\

\noindent
{\bf st-order:} \\
Let $v_1,v_2,\cdots, v_n$ be an st-order of $G$. This order is characterized by the fact that at any stage $i \in \brac{1,\cdots, n-1}$ not just the non-embedded vertices but also the embedded vertices $\brac{v_1,\cdots, v_i}$ induce a connected component. This results in the algorithm having to maintain only one PQ-tree $T_i$ to capture all possible circular orderings of out-edges around this one component. Furthermore since $v_1v_n$ is an out-edge at every stage, it can stay as the special leaf of $T_i$ for all $i$.

The update step for embedding the next vertex is the same as above. At stage $1$, the tree $T_1$ consists of the special leaf $v_1v_n$ and a P-node whose children are all other out-edges of $v_1$. Now suppose we are at a stage $i \in \brac{1,\cdots,n-2}$.  Again it holds that after embedding $v_{i+1}$ the edges incident to $v_{i+1}$ have to be consecutive. We call the set of leaves of $T_i$ that correspond to these edges the {\em black} leaves. To go to the next stage, we thus first reduce $T_i$ so that all black edges appear together. A non-leaf node in the reduced PQ-tree is said to be black if all its descendants are black edges. We next create a new P-node $p_{i+1}$ and add all the out-edges of $v_{i+1}$ as its children. Now $T_{i+1}$ is constructed from $T_i$ as follows:

\noindent
{\bf Case 1:} {\it $T_i$ contains a black node $x$ that is an ancestor of all the black leaves}. 
  We obtain $T_{i+1}$ from $T_i$ by replacing $x$ and all its descendants with $p_{i+1}$. 
  
\noindent
{\bf Case 2:} {\it $T_i$ contains a (non-black) Q-node containing a (consecutive) sequence of black children}. 
  We obtain $T_{i+1}$ from $T_i$ by replacing these black children (and their descendants) with $p_{i+1}$. 
  
This captures again exactly all possible circular orders of half-embedded edges around the newly embedded connected component. Note that as before the graph is non-planar if and only if the reduction step fails at any stage.\\

\section{Simultaneous planarity when the common graph is 2-connected} 
\label{sec:simplanarity}

This section contains our main result, a linear time algorithm for testing simultaneous planarity of two graphs that share a 2-connected subgraph.

Let $G_1=(V_1,E_1)$ and $G_2=(V_2,E_2)$ be two planar connected graphs with $|V_1| = n_1$ and $|V_2|=n_2$.
Note that $G_1[V_1 \cap V_2]$ need not be the same as $G_2[V_1 \cap V_2]$.  
Let $G=(V_1 \cap V_2, E_1 \cap E_2)$ be 2-connected and $n = |V_1 \cap V_2|$.
Let $v_1,v_2,\ldots, v_n$ be an st-ordering
of $V_1 \cap V_2$. We call the edges and vertices of $G$ {\em common} and all other vertices and edges {\em private}.

We say two linear or circular orderings of elements with some common elements are {\em compatible} if the common elements appear in the same relative order in both orderings. Similarly we say two combinatorial planar embeddings of $G_1$ and $G_2$ respectively are {\em compatible} if for each common vertex the two circular orderings of edges incident to it are compatible.

If $G_1$ and $G_2$ have simultaneous planar embeddings, then clearly they have combinatorial planar embeddings that are compatible with each
other. The converse also turns out to be true, if the common edges form a connected graph. This can be easily proved
as follows. Let ${\cal E}_1$ and ${\cal E}_2$ be the compatible combinatorial planar embeddings of $G_1$ and $G_2$
respectively. Let ${\cal E}_p$ be the partial embedding of ${\cal E}_1$ [or ${\cal E}_2$] obtained by restricting $G_1$ [resp. $G_2$]
to the common subgraph. (Note that since ${\cal E}_1$ and ${\cal E}_2$ are compatible, the partial embedding of ${\cal E}_1$ restricted to
the common subgraph is the same as the partial embedding of ${\cal E}_2$ restricted to the common subgraph.) Now we can find a planar
embedding of ${\cal E}_p$ and iteratively extend it to an embedding of ${\cal E}_1$ and an embedding of ${\cal E}_2$ (see Lemma~2 of
J{\"u}nger and Schultz~\cite{JS} for a proof). The two planar embeddings thus obtained are simultaneous planar embeddings and thus it
is enough to compute a pair of compatible combinatorial planar embeddings.

\begin{comment}
They show that given a planar embedding $\cal E$ of a planar graph $H$, and given a partial planar embedding ${\cal E}_p$ obtained
(In fact, they prove a stronger
theorem that holds for compatible topological planar embeddings of possibly disconnected graphs. A topological planar embedding of a graph is
a combinatorial planar embedding together with information about the relative positions of the connected components, i.e.~it contains
information about which component is contained in which face of the graph).
\end{comment}

We will find compatible combinatorial planar embeddings by adding vertices one by one.
%, iteratively constructing two sets of PQ-trees, representing the partial planar embeddings of $G_1$ and of $G_2$ respectively.
At any point we will have two sets of PQ-trees %, one set for $G_1$ and one for $G_2$, 
representing the partial planar embeddings  
of the subgraphs of $G_1$ and $G_2$ induced by the vertices added so far.
Each PQ-tree represents one connected component of the current subgraph of $G_1$ or $G_2$.
In the first phase we will add  
all private vertices of $G_1$ and $G_2$,
and in the second phase we will add the common vertices in an st-order.  
When a common vertex is added, it will appear in two PQ-trees, one for $G_1$ and one for $G_2$ and we must take care to maintain compatibility. Note that if $V_1=V_2$, i.e., when all vertices are common, only the second phase is needed and the whole algorithm will solely operate on two PQ-trees -- one for each graph. 
% During the second phase we must take care to maintain compatibility.

% {\ALnote we can say something high-level about orientation variables here}

Before describing the two phases, we give the main idea of maintaining compatibility between two PQ-trees.  
Recall that in Section~\ref{sec:int-PQ} we found compatible orders for two PQ-trees using
projection and intersection of PQ-trees.  However, we were unable to store a set of compatible orderings as a PQ-tree, which is what we really need, since planarity testing involves building a sequence of PQ-trees as we add vertices one by one.

To address this issue we introduce a boolean ``orientation''  variable attached to each Q-node to encode whether it is ordered forward or backward.  Compatibility is captured by equations relating orientation variables.
%When a Q node of the first PQ-tree and a Q node of the second PQ-tree must be compatibly ordered, we equate their orientation variables.
At the conclusion of the algorithm, it is a simple matter to see if the resulting set of Boolean equations has a solution. If it does, we use the solution to create compatible orderings of the Q-nodes of the two PQ-trees. Otherwise the graphs do not have simultaneous planar embeddings.

In more detail, we create a Boolean orientation variable $f(q)$ for each Q-node $q$, with the interpretation that $f(q) = {\rm true}$ if and only if $q$ has a
``forward'' ordering. We record the initial ordering of each Q-node in order to distinguish ``forward'' from ``backward''.
During PQ-tree operations,
Q-nodes may merge, and during planarity testing, parts of PQ-trees may be deleted.   We handle these modifications to Q-nodes by the simple expedient of having an orientation variable for each Q-node, and equating the variables as needed.
When Q-nodes $q_1$ and $q_2$ merge, we add the equation $f(q_1) = f(q_2)$ if $q_1$ and $q_2$ are merged in the same order (both forward or both backward), or $f(q_1) = \neg f(q_2)$ otherwise.

We now describe the two phases of our simultaneous planarity testing algorithm.
% {\ALnote end of added section}
To process the private vertices of $G_1$ and $G_2$ in the first phase we compute for each graph a reverse depth-first search ordering by contracting $G$ into a single vertex and then running a depth-first search from this vertex. With these orderings we can now run the algorithm of Haeupler and Tarjan for all private vertices as described in Section~\ref{sec:planar}.

Now the processed vertices induce a collection of components, such that each component has an out-edge to a common vertex. Further, the planarity test provides us for each component with an associated PQ-tree representing all possible cyclic orderings of out-edges for that component. For each component we look at the out-edge that goes to the first common vertex in the st-order and re-root the PQ-tree for this component to have this edge represented by the special leaf. This completes the first phase.

For the second phase we insert the common vertices in an st-order. The algorithm is similar to that described in Section~\ref{sec:planar} for an st-order but, in addition, has to take care of merging in the private components as well. We first examine the procedure for a single graph. Adding the first common vertex $v_1$ is a special set-up phase which we describe first; we will describe the general addition below. 

Adding $v_1$ joins some of the private components into a new component $C_1$ containing $v_1$.
 For each of these private components we reduce the corresponding PQ-tree
so that all the out-edges to $v_1$ appear together, and then delete those edges.
Note that due to the re-rooting at the end of the first phase the special leaf is among those edges. Thus the resulting PQ-tree represents the linear orderings of the remaining edges. We now build a PQ-tree representing the circular orderings around the new component $C_1$ as follows: we take $v_1 v_n$ as the special leaf, create a new P-node as a root and
add
all the out-edges of $v_1$ and the roots of the PQ-trees of the merged private components
as children of the root
(see Figure~\ref{FIG:setup}).

\begin{figure}[htbp]
   \centering
   \includegraphics[width=2.2in]{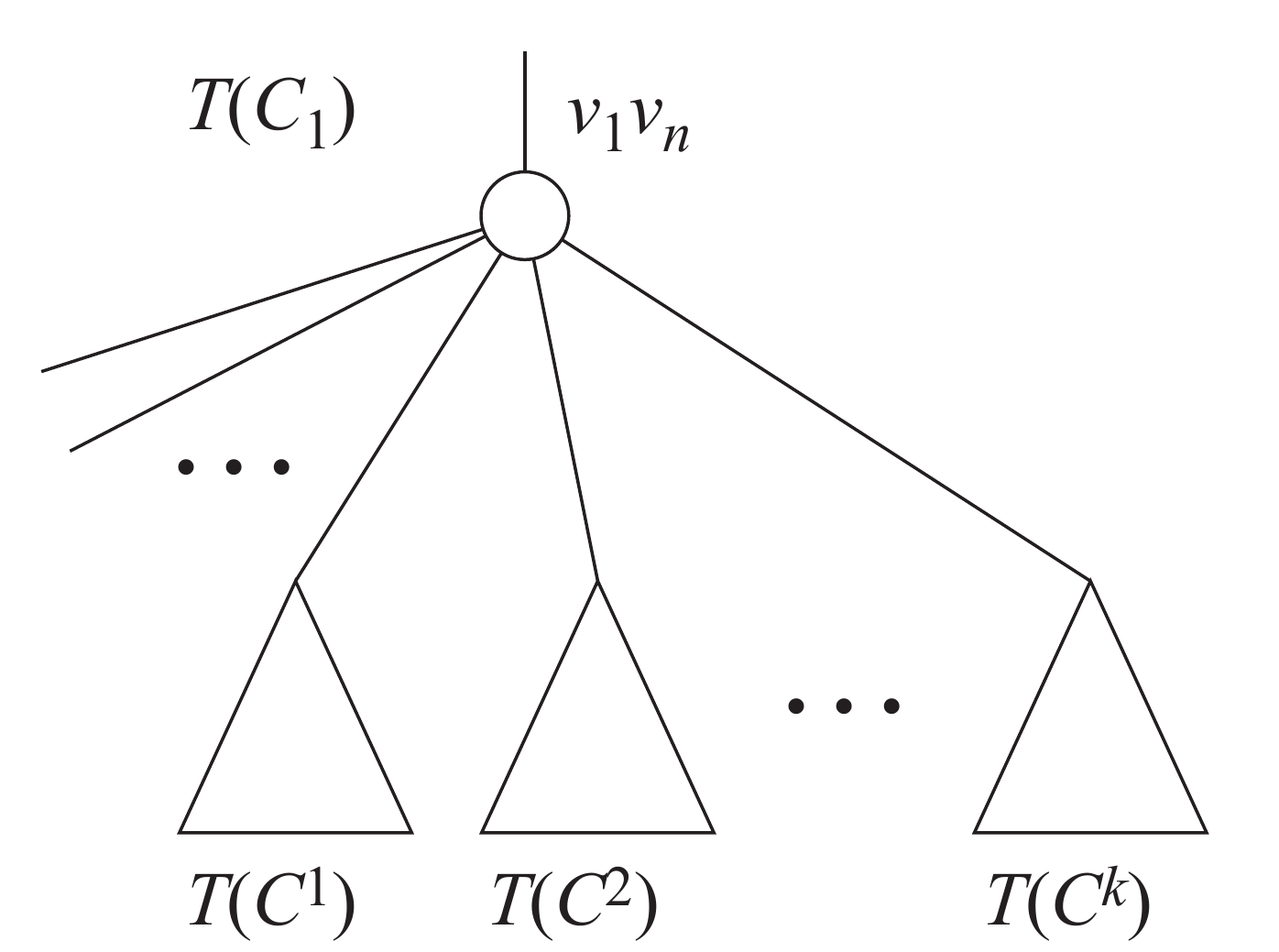}
   \caption{Setting up $T(C_1)$. The P-node's children are the outgoing edges of $v_1$ and the PQ-trees for the components that are joined together by $v_1$.}
   \label{FIG:setup}
\end{figure}

% AL.  Bernhard took out all notation but I will add back a bit -- this is the quickest way to merge his writing into Raju's
Now consider the situation when we are about to add the common vertex $v_i$, $i \ge 2$. The graph so far may have many connected components but because of the choice of an st-ordering, all common vertices embedded so far are in one component $C_{i-1}$, which we call the \emph{main} component. When we add $v_i$, all components with out-edges to $v_i$ join together to form the new main component $C_i$. This includes $C_{i-1}$ and possibly some private components.
%Let ${\cal C}_i$ be this set of private components.
The other private components do not change, nor do their associated PQ-trees.
 
%%% BERNHARD: This is how far I came %%%

% here I am continuing Bernhard's notation-elimination  
We now describe how to update the PQ-tree $T_{i-1}$ associated with $C_{i-1}$ to form the PQ-tree $T_i$ associated with $C_i$. This is similar
to the approach described in Section~\ref{sec:planar}.  
%The leaves of $T_{i-1}$ corresponding to edges incident to $v_i$ are called \emph{black} leaves.
 We first reduce $T_{i-1}$ so that all the \emph{black} edges (the ones incident to $v_i$) appear together. As before, we call a non-leaf node
in the reduced PQ-tree \emph{black} if all its descendants are black leaves.
For any private component with an out-edge to $v_i$, we reduce the corresponding PQ-tree so that all the out-going edges to $v_i$ appear together,
and then delete those edges.  We make all the roots of the resulting PQ-trees into children of a new P-node $p_i$, and also add all the out-going
edges of $v_i$ as children of $p_i$. It remains to add $p_i$ to $T_{i-1}$ which we do as described below. In the process we also create a
\emph{black tree} $J_i$ that represents the set of linear orderings of the black edges.

\noindent
{\bf Case 1:} {\it $T_{i-1}$ contains a black node $x$ such that all black edges are descendants of $x$}. Let $J_i$ be the subtree rooted at $x$. We obtain $T_i$ from $T_{i-1}$ by replacing $x$ and all its descendants with $p_i$.

\noindent
{\bf Case 2:} { {\it $T_{i-1}$ contains a non-black Q-node $x$ that has a sequence of adjacent black children}.
 We group all the black children of $x$ and add them as children (in the same order) of a new Q-node $x'$. Let
$J_i$ be defined as the subtree rooted at $x'$. We add an equation relating the orientation variables of $x$ and $x'$.
%define the literal of $x'$ to be the same as the literal of $x$.
% I think we can defer markers til the unification step
%
% Observe that in any circular ordering of $T(C_{i-1})$ the Q-nodes
% $x$ and $x'$ must be oriented in the same way.
% Thus we define the literal of $x'$ to be the same as the literal of $x$.
We obtain $T_i$ from $T_{i-1}$ by replacing the sequence of black children of $x$ (and their
descendants) with $p_i$ (see Figure~\ref{FIG:step}).

\begin{figure}[htb]
   \centering
   \includegraphics[width=5in]{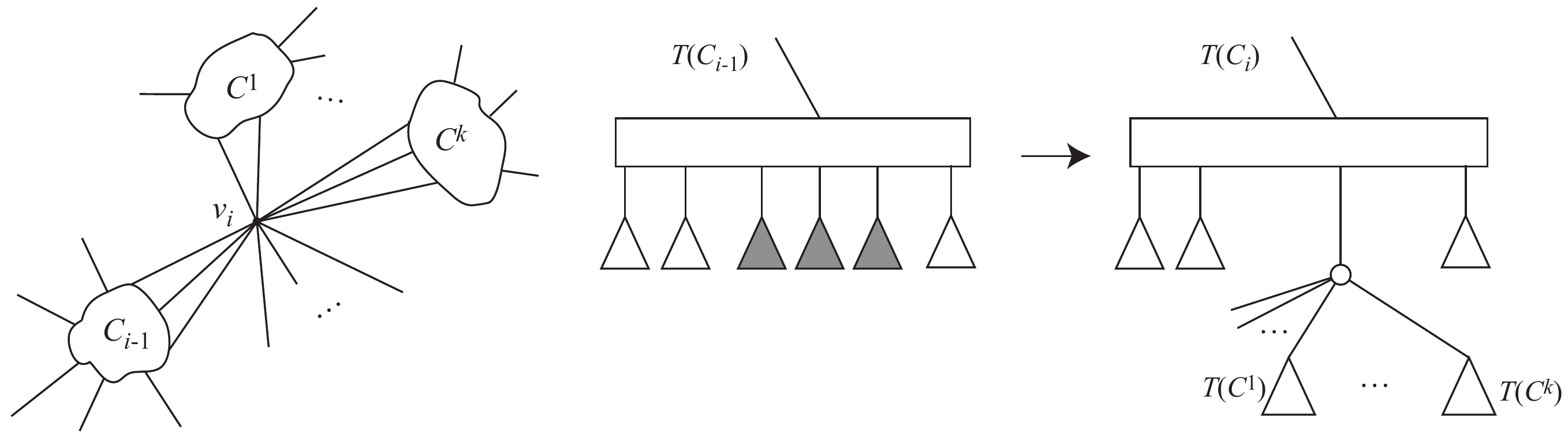}
   \caption{({\it left\/}) Adding vertex $v_i$ which is connected to main component $C_{i-1}$ and to private components $C^1, \ldots ,C^k$.  ({\it right\/}) Creating $T_i$ from $T_{i-1}$ by replacing the black subtree by a P-node whose children are  the outgoing edges of $v_i$ and the
   PQ-trees for the newly joined private components.}
\label{FIG:step}
\end{figure}

Note that we use orientation variables above for a purpose other than compatibility. (We are only working with one graph so far.)  Standard planarity tests would simply keep track of the order of the deleted subtree $J_i$ in relation to its parent.  
Since we have orientation variables anyway, we use them for this purpose.

We perform a similar procedure on graph $G_2$.   We will distinguish the black trees of $G_1$ and $G_2$ using superscripts.  Thus after adding $v_i$  we have black trees $J_i^1$ and $J_i^2$.
It remains to deal with compatibility.  
We claim that it suffices to enforce compatibility between each pair $J_i^1$ and $J_i^2$.

To do so, we perform a {\em unification step} in which we add equations between orientation variables for Q-nodes in the two trees.  \\

\noindent
\textbf{Unification step for stage $i$}

We first project $J_i^1$ and $J_i^2$ to the common edges, as described in Section~\ref{sec:int-PQ}, carrying over orientation variables from each original node to its copy in the projection (if it exists).
Next we create the PQ-tree $R_i$ that is the intersection of these two projected trees as described in Section~\ref{sec:int-PQ}.  Initially $R_i$ is equal to the first tree.  The step dealing with Q-nodes (Step 3) is enhanced as follows:

\begin{description}
\item{3.} For each Q-node $q$ of the second tree, and for each pair $a_1$, $a_2$ of adjacent children of $q$ do the following:  Reduce $R_i$ by adding a consecutivity constraint on all the descendant leaves of $a_1$ and $a_2$.  Find the Q-node that is the least common ancestor of the descendants of $a_1$ and $a_2$ in $R_i$.  Add an equation relating the orientation variable of this ancestor with the orientation variable of $q$ (using a negation if needed to match the orderings of the descendants).

\end{description}

Observe that any equations added during the unification step are necessary.   Thus if the system of Boolean equations is inconsistent at
the end of the algorithm, we conclude that $G_1$ and $G_2$ do not have a compatible combinatorial planar embedding.  
% cheating here too
%%%%%%%%%%%%%%%%%
%%%%%%%%%%%%%%%%%%%
Finally, if the system of Boolean equations has a solution, then we obtain compatible leaf-orders for each pair $J_i^1$ and $J_i^2$ as follows:
Pick an arbitrary solution to the system of Boolean equations. This fixes the truth values of all orientation variables and thus the orientations
of all Q-nodes in all the trees.
Subject to this, choose a leaf ordering
%Choose an arbitrary leaf-order
$I$ of $R_i$ (by choosing the ordering of any P-nodes).
$I$ can then be lifted back to (compatible) leaf-orders
of $J_i^1$ and $J_i^2$ that respect the ordering of $I$.
The following lemma shows that this is sufficient to obtain compatible combinatorial planar embeddings of $G_1$ and $G_2$

\begin{lemma}
\label{compatibility}
If the system of Boolean equations has a solution then $G_1$ and $G_2$ have compatible combinatorial planar embeddings.
%If there are compatible leaf-
%For $j=\brac{1,2}$ and $i=\brac{2,\cdots,n}$, let $I_i^j$ be a leaf-order of $J_i^j$ (with respect to %some solution of the system of Boolean equations).
%If $I_i^1$ is compatible with $I_i^2$ for all $i$, then $G_1$ and $G_2$ have compatible %combinatorial embeddings.
\end{lemma}

\begin{proof}
The procedure described above produces compatible leaf orders for all pairs of black trees $J_i^1$ and $J_i^2$. Recall that
the leaves of $J_i^1$ (resp. $J_i^2$) are the out-edges of the component $C_{i-1}$ in $G_1$ (resp. $G_2$) and contain all the
common in-edges of $v_i$. Focusing on $G_1$ individually, its planarity test has succeeded, and we have
a combinatorial planar embedding such that the ordering of edges around $v_i$ contains the leaf order of $J_i^1$. Also,
we have a combinatorial planar embedding of $G_2$ such that the ordering of edges around $v_i$ contains the leaf order of $J_i^2$.

The embedding of a graph imposes an ordering of the out-edges around every main component. We can show inductively, starting from $i=n$, that
the ordering of the out-edges around the main component $C_{i-1}$ in $G_1$ is compatible with the ordering of the out-edges in the corresponding
main component in $G_2$. Moreover all the common edges incident to $v_i$ belong to either $C_{i-1}$ or $C_i$. This implies that in both embeddings,
the orderings of edges around any common vertex are compatible. Therefore $G_1$ and $G_2$ have compatible combinatorial planar embeddings.
\end{proof}

\subsection{Simultaneous planarity of $k$ graphs}

	In this subsection we consider a generalization of simultaneous planarity for $k$ graphs, when each vertex [edge] is either present in all
the graphs or present in exactly one of them. Let $G_1 = (V_1, E_1), G_2 = (V_2, E_2), \cdots, G_k = (V_k, E_k)$ be $k$ planar graphs such
that $V = V_i \cap V_j$ for all distinct $i,j$ and $E = E_i \cap E_j$ for all distinct $i,j$. As before, we call the edges and vertices of
$G=(V,E)$ {\em common} and all other edges and vertices {\em private}. We show that the algorithm of Section~\ref{sec:simplanarity}
can be readily extended to solve this generalized version, when the common graph $G$ is 2-connected.

  If $G_1, G_2, \ldots, G_k$ have simultaneous planar embeddings then they clearly have mutually compatible combinatorial planar embeddings.
Conversely if $G_1, G_2, \ldots, G_k$ have combinatorial planar embeddings that are mutually compatible, then, as before (see the beginning
of Section~\ref{sec:simplanarity}), we can first find
the planar embedding of the common subgraph and extend it to the planar embeddings of $G_1, G_2, \ldots, G_k$. Thus once again, the problem
is equivalent to finding combinatorial planar embeddings for $G_1, G_2, \ldots G_k$, that are mutually compatible.

Our algorithm for finding such an embedding works as before, inserting private vertices first, followed by common vertices.
The only difference comes in the unification step, where we have to take the intersection of $k$ projected trees instead of $2$. Doing this
is straightforward:
We initialize the intersection tree to be the first projected tree, and then insert the constraints of all the other trees into the intersection tree.
%We start with the first projected tree as the intersection tree and for every other tree, we insert its constraints into the
%intersection tree.
Finally, Lemma~\ref{compatibility} and its proof extend to multiple graphs.

\subsection{Running time}

%So far we have given an algorithm to test if two graphs that share a 2-connected subgraph have simultaneous planar embeddings. It is straight forward to implement the given algorithm in polynomial time and for sake of clarity we have not included the details.
%Nevertheless we are convinced that using the classical techniques for efficiently implementing PQ-trees and PQ-tree based planarity tests one can obtain a linear-time algorithm:

%Assuming that the input graphs are specified explicitly (common vertices/edges appear once for each graph), we show that our algorithm can be
%implemented in linear time:
We show that our algorithm can be implemented to run in linear time.  (In the generalization to $k$ graphs, the run time is linear in the sum of the sizes of the graphs---in other words, a common vertex counts $k$ times.)
Computing the reverse depth-first search ordering and the st-ordering are known to be 
%feasible 
doable in linear time \cite{ET}.
The first phase of our algorithm uses PQ-tree based planarity testing with a reverse depth-first search order~\cite{HT}, which
% The first phase of the planarity-test on the private vertices uses the reverse depth-first ordering and
runs in linear time using the
efficient PQ-tree implementation of Booth and Lueker~\cite{Booth,BL}.
The re-rooting between the two phases needs to be done only once and
can easily be done in linear time.  
The second phase of our algorithm uses PQ-tree based planarity testing with an st-order, as discussed in Section~\ref{sec:planar}. 
This avoids re-rooting of PQ-trees, and thus also runs in
% The planarity test in the second phase is performed on an st-order. This is essentially the only other
% order that avoids re-rooting of PQ-trees during the planarity test and thus also runs in
linear time~\cite{HT,BL,LEC}.  The other part of the second phase is the unification step, which
is only performed on the black trees, i.e.~the edges connecting to the current vertex. Note that these edges will get deleted and will not appear in subsequent stages. Thus we can explicitly store the black trees and the
intersection tree at every stage and allow the unification step to take time linear in the 
%complete 
total size of both black trees.
The intersection algorithm can be implemented in linear time, as mentioned in Section~\ref{sec:int-PQ}.
The last thing that needs to
be implemented efficiently is the handling of the orientation variables. It is easy to see that once the equations are generated, they can be
solved in linear time, by repeatedly fixing the value of a free variable 
%to be true or false, 
(true or false), finding all the equations that contain the variable
and recursively (say in a depth-first way) fixing all the variables so as to satisfy the equations. In the following sub-section we explain how
to also generate the variable equations in linear time. With this we conclude that the algorithm runs in linear time.

% With Q-nodes (in the efficient standard implementation~\cite{Booth,BL}) being only implicitly represented as undirected doubly linked lists of its children this is again not completely trivial. Storing variable information at the current place of the doubly linked list and checking contradictions for example only in the end resolves this problem, too.

\subsubsection{Generating variable equations in linear time}
Note that in the implementation of PQ-trees (see Booth and Lueker~\cite{BL}) the children of a Q-node are stored in a doubly-linked list and only
the leftmost and rightmost children have parent pointers. Thus when two Q-nodes, one a child of the other, merge, we may not know
the variable and the orientation of the parent Q-node. To address this problem, we use labels on certain links of the doubly-linked
list as explained below and compute all the equations generated by the reductions of a unification step at the end of the step.

For any two adjacent child nodes $c_i$ and $c_{i+1}$ of a Q-node $q$, either the links $c_i \rightarrow c_{i+1}$ and $c_{i+1} \rightarrow c_i$
are labeled with $l$ and $\neg l$ (respectively), for some literal $l$, or they are both unlabeled. The underlying interpretation is that
$c_i$ appears before $c_{i+1}$ in the child ordering of $q$ iff $l$ is true. Thus the literals that we encounter when traveling from
one end to the other of the doubly-linked list, are all (implicitly) equal. When a Q-node is first created with two child nodes, say $x$ and $y$,
we create a variable associated with it and label the link from $x$ to $y$ with the variable and the link from $y$ to $x$ with the negation
of the variable.

During the algorithm, there are two types of equations: (1) Equations consisting of literals appearing in Q-nodes of distinct trees.
(These can be PQ-trees of the same graph, as happens in Case 2 of Section~\ref{sec:simplanarity} or PQ-trees of different graphs, as
happens in step 3 of Unification.)
(2) Equations consisting of literals
appearing in Q-nodes of a single tree (created during PQ-tree reductions).

Note that type 1 equations are essentially equations that constrain
the ordering of child nodes across the two Q-nodes, and we handle them as follows. Let $c_1, c_2$ be any two adjacent child nodes
of the first Q-node that are constrained to appear in the same order as child nodes $c'_1, c'_2$ of the second Q-node. If the links between
$c_1$ and $c_2$ are unlabeled, we create a new variable $x$ and label the links $c_1 \rightarrow c_2$ and $c_2 \rightarrow c_1$ with
$x$ and $\neg x$ respectively. Similarly, we label the links between $c'_1$ and $c'_2$, if they are unlabeled. Now we create the equation
that equates the literal associated with $c_1 \rightarrow c_2$ with the literal associated with $c'_1 \rightarrow c'_2$.

Type 2 equations happen when two Q-nodes merge. In this case the merged node contains literals from both the Q-nodes. Finding an equation
after each merge is costly, as we need to scan through each Q-node until we find a literal. However we can compute the equations in a
lazy fashion at the end of each unification step as follows. For every Q-node of the two black trees and the intersection tree obtained
from their projections (i.e.~the output tree of the unification step on the black trees), we pass from the first child to the last child
and equate all the literals encountered in the labels of the links. This clearly takes linear time in the size of the black trees.

\subsection{Simultaneous planar drawings}
%\section{Constructing planar embeddings from compatible combinatorial embeddings}
\label{sec:drawing}

% I re-wrote this section to be more consistent with the style of the rest of the paper.  E.g. our main result is not presented as a theorem, so this result should not either.  
% I also made it a subsection, rather than a top level section

Our algorithm to test simultaneous planarity finds 
compatible combinatorial planar embeddings but does not find actual simultaneous drawings.  
In this section we make some remarks on finding simultaneous planar drawings.

Two natural goals are to minimize the number of bends along the polygonal curves representing the edges, and to minimize the number of crossings between the private edges of the two graphs.
We will show that we can find drawings where each edge has at most $n$ bends, and each pair of private edges intersects at most once.  Recall that $n = |V_1 \cap V_2|$.  In fact any edge of $G_1$, and any private edge of $G_2$ that joins private vertices of $G_2$ will be drawn as straight line segments.

Take a straight-line planar drawing of the first graph, $G_1$.
%, with the outer face convex, and add edges to triangulate every interior face in the drawing.  The resulting graph, $G_1'$, has size $O(n_1)$ and its drawing 
This includes a straight-line drawing of the common graph $G$, which we will extend to a planar drawing of $G_2$. 

The main idea is to draw the private parts of $G_2$ inside the faces of the drawing of $G$ using shortest paths.
(Of course, we must offset the shortest paths slightly to avoid coincident features.) 
A shortest path inside a polygon has two crucial properties: it bends only at reflex vertices of the polygon, and it intersects any line segment contained in the polygon at most once.  

Any bend on an edge of $G_2$ will occur in the neighborhood of a common vertex (a vertex of $G$).
We will maintain the following invariants:
(1) each edge of $G_2$ has at most one bend per common vertex;  
(2) the boundary of each face of $G_2$ has at most one reflex bend per common vertex; and 
(3) any edge of $G_1$ intersects any face of $G_2$ in a single line segment.
Property (1) implies that each edge of $G_2$ has at most $n$ bends.  Property (2) allows us to maintain property (1) as we draw more of $G_2$.   Property (3) ensures that private edges of $G_1$ cross private edges of $G_2$ at most once.
Initially, the only part of $G_2$ that is drawn is $G$, and the properties are satisfied.

\begin{figure}[htbp]
   \centering
   \includegraphics[width=6.5in]{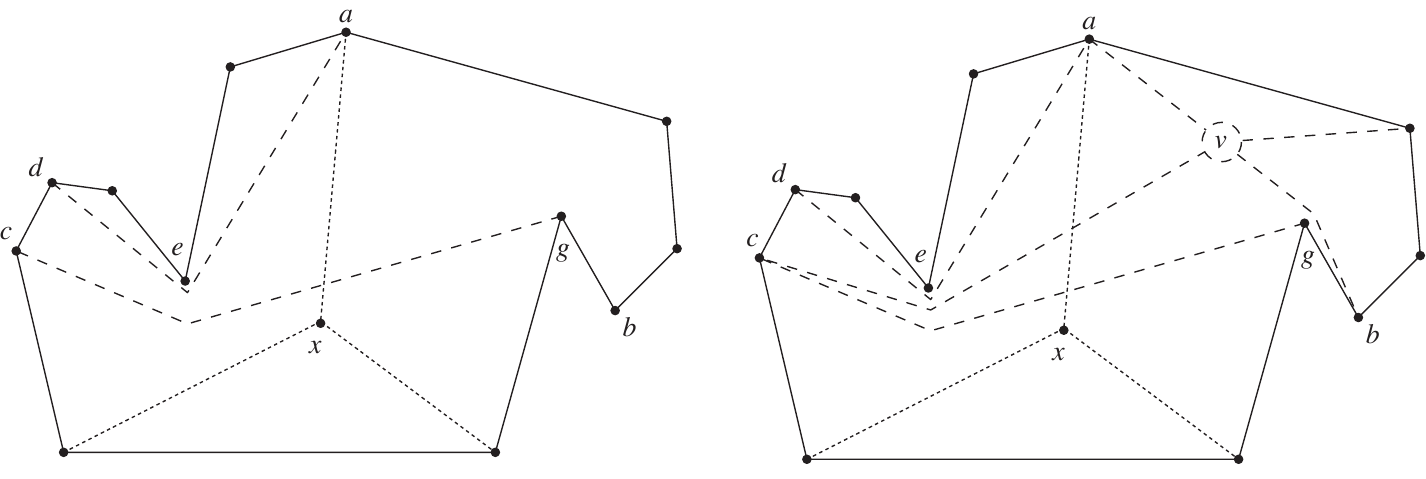}
   \caption{({\it left\/}) A face of $G$ (drawn with solid lines), some private edges of $G_1$ (dotted lines), and two private edges of $G_2$ (dashed lines) with endpoints in $G$: $(d,a)$ and $(c,g)$, each drawn with a bend in the neighborhood of the common vertex $e$.  Observe that although there are two bends near $e$, each face of $G_2$ has only one reflex bend near $e$. 
    ({\it right\/})  Adding private vertices of $G_2$: vertex $v$ is drawn along a shortest path from $a$ to $b$ and a component is drawn in a small disc around $v$.}
   \label{FIG:sim-draw}
\end{figure}

\begin{comment}
We will maintain the invariant that the bends of the edges of $G_2$ are associated with the common vertices (the vertices of  $G$), 
in such a way that: (1) each edge of $G_2$ has at most one bend per common vertex;  and (2) the boundary of each face of $G_2$ has at most one reflex bend per common vertex.   Furthermore, every bend will be very close to its associated common vertex.  Property (1) implies that each edge of $G_2$ has at most $n$ bends, and property (2) allows us to maintain property (1) as we draw more of $G_2$.  Keeping bends close to the common vertices ensures that 
Initially, we only have a drawing of $G$ and both properties are satisfied.
\end{comment}

We first draw the private edges of $G_2$ whose endpoints are in $G$.  
Consider one such edge.  The combinatorial embedding specifies  the face, $f$, of the partial embedding of $G_2$ that contains the edge.
Face $f$ is drawn as a polygon and, by property (3), any  reflex vertex of the polygon is in the neighborhood of a unique common vertex. Draw the edge as a shortest path in $f$, pulling it slightly away from the boundary of $f$.  
Observe that this satisfies our invariants.

We now insert private vertices of $G_2$ and their incident edges.  
Consider a connected component $C$ of 
$G_2 - (V_1 \cap V_2)$, consisting of some private vertices and their induced edges (call this $C'$) together with some edges joining $C'$ to the common graph $G$.
The combinatorial embedding specifies the face, $f$, of the partial embedding of $G_2$ that contains $C$.   
Face $f$ is drawn as a polygon and, by property (3), any  reflex vertex of the polygon is in the neighborhood of a unique common vertex. 
We will contract $C'$ to a single vertex $v$ and draw $v$ (as a point) together with all the edges joining it to common vertices.  After that, we take a small disc around the point, small enough to avoid all edges of $G_1$,  and we make a straight line drawing of $C'$ inside the disc.  
Thus it suffices to consider adding a single private vertex $v$ and the edges joining it to common vertices.  
If $v$ is adjacent to only one common vertex, we can augment $G_2$ to make it adjacent to another common vertex.  
Let  $a$ and $b$ be two common vertices  adjacent to $v$.  Draw a shortest path inside $f$ from $a$ to $b$, and choose some point along that path to place $v$.  Draw the remaining edges incident to $v$ as shortest paths in the polygon, pulling the paths slightly away from the boundary of $f$ and from each other to avoid coincident features.  Observe that $v$ is not a reflex vertex in any of the resulting faces.
This construction satisfies our invariants: each edge has at most one bend per reflex vertex of $f$, the boundary of each resulting face has at most one reflex bend per vertex of $f$, and any edge of $G_1$ intersects any face of $G_2$ in a single line segment.

\section{Conclusions}
%\section{Summary and Open Problems}
We have given a linear-time algorithm to test simultaneous planarity of two graphs 
%that share a 2-connected subgraph. 
whose common subgraph is 2-connected.
Our algorithm does not require the two graphs to have the same vertex set. 
Furthermore, our algorithm works for the more general case when there are $k$ graphs, any two of which 
%share the same 2-connected subgraph. 
have the same 2-connected subgraph in common.

%Very recently Bl{\"a}sius and Rutter~\cite{BR} have given a linear time algorithm to 
%to test simultaneous planarity of two biconnected graphs that share a connected subgraph.  They use simultaneous PQ trees, but in a more general formulation than ours.
 
% In our algorithm, 2-connectivity restricts each graph to have at most
% one main component at any stage, thus making the unification step simpler.

% A big open question is to determine whether simultaneous planarity of two graphs can be tested in polynomial time. %We conjecture that
%the problem can be solved in polynomial time when the common graph is 1-connected. Note that in this case the %problem is still equivalent
%to finding compatible combinatorial planar embeddings for the two graphs.

We conjecture that simultaneous planarity of multiple graphs can be tested in polynomial time when any vertex/edge is either private to a single graph or common to all graphs.  Our algorithm solves the special case when the common graph is 2-connected.  
Note that simultaneous planarity is known to be NP-hard for three graphs  in general~\cite{GJPSS}.
The three graphs constructed in the reduction have edges in all possible combinations: edges private to each graph, edges common to all three graphs, and edges common to a pair of graphs but not the third, for every pair.   

Another interesting case of simultaneous planarity  for multiple graphs is when the graphs intersect in {\em layers}: the graphs are ordered $G_1, \ldots, G_k$, and every vertex/edge occurs in a consecutive subsequence of the graphs.  This layered structure arises from a graph changing over time, so long as a vertex/edge does not re-appear after it vanishes.   The complexity of layered simultaneous planarity is open, even for three graphs.  
For three graphs the layered structure is equivalent to the the condition that if a vertex/edge is in $G_1$ and $G_3$ then it is also in $G_2$.  

A weaker version of simultaneous planarity for two graphs requires only that each common vertex be represented by the same point---a common edge may be represented by different curves in the two planar drawings.  
Any pair of planar graphs can be represented this way:
after choosing arbitrary vertex positions, each graph can be drawn independently,  as shown by Pach and Wenger~\cite{PW}.
It would be interesting to 
generalize our algorithm to the case where some designated common edges are allowed to be represented by different curves in the two drawings.  
There is also a natural optimization version of the problem: given two planar graphs with some common vertices and edges, find planar drawings of the graphs so that every common vertex is drawn as the same point in the two drawings, and maximize the number of common edges that are drawn as the same curve in the two drawings.

Lastly, both SPQR-trees~\cite{BT} and PQ-trees have been used for many planarity related problems and both give rise to distinct representations of all planar embeddings of a (2-connected) planar graph~\cite{BT,CNSO,HT}. Understanding the strength of both representations and how they relate better is an important question. Comparing the PQ-tree approach to simultaneous planarity taken here with the SPQR-tree approach in \cite{ADFPR} or \cite{FGJMS} might be a good start. The recent manuscript \cite{BR} makes very interesting progress in this direction.

\bibliography{report.bib}

\end{document}